\documentclass[12pt]{amsart}

\usepackage{color}

\usepackage{graphicx}

\usepackage{amsmath,amsthm}
\makeatletter
\def\amsbb{\use@mathgroup \M@U \symAMSb}
\makeatother

\usepackage{bbold}

\usepackage{amssymb}
\usepackage{caption}
\usepackage[latin1]{inputenc}
\usepackage{enumerate}
\usepackage{mathrsfs}  % Para las fuentes mathscr
\usepackage{dsfont}

\usepackage[bbgreekl]{mathbbol}

\usepackage{geometry}
\usepackage[all]{xy}

\newtheorem{thm}{Theorem}[section]

\newtheorem{prop}[thm]{Proposition}

\theoremstyle{definition}
\newtheorem{defi}[thm]{Definition}
\newtheorem{obs}[thm]{Remark}

\numberwithin{equation}{section}

\def\C{\mathcal {C}}
\def\Com{\mathbb{C}}
\def\A{{\mathcal{C}}^\infty(M)}
\def\AT{{\mathcal{C}}^\infty(TM)}
\def\AC{{\mathcal{C}}^\infty(T^*M)}

\def\F{\mathcal {F}}

\def\Ocal{\mathcal{O}}
\def\R{\amsbb{R}}%{\mathbb R}
\def\m{\mathfrak m}

\def\To{\longrightarrow}

\def\PHI{\boldsymbol{\Phi}}

\def\exp{\textrm{exp}\,}

\def\be{\begin{equation}
}

\def\ee{\end{equation}
}

\begin{document}

\title[Riemannian exponential and quantization]
{Riemannian exponential and quantization}

\author{ J.  Mu\~{n}oz-D{\'\i}az and R. J.  Alonso-Blanco}

\address{Departamento de Matem\'{a}ticas, Universidad de Salamanca, Plaza de la Merced 1-4, E-37008 Salamanca,  Spain.}
\email{ricardo@usal.es}

\begin{abstract}
This article continues and completes the previous one \cite{QuantizacionMunozAlonso}. First of all, we present  two methods of quantization associated with a linear connection given on a differentiable manifold, one of them being the one presented in \cite{QuantizacionMunozAlonso}. The two methods allow quantize functions that come from covariant tensor fields. The equivalence of both is demonstrated as a consequence of a remarkable property of the Riemannian exponential (Theorem \ref{tidentidad1}) that, as far as we know, is new to the literature. On the other hand, the extension of the previously mentioned quantization to functions of a very broad type can be carried out by generalizing the method of \cite{QuantizacionMunozAlonso} in terms of fields of distributions.
\end{abstract}
\bigskip

\maketitle

%\centerline{\today}

\tableofcontents

\setcounter{section}{-1}

\bigskip

\section{Introduction}

The ``factor-ordering problem'' was unavoidable within the Matrix Mechanics methods. It has been kept in Quantum Mechanics textbooks as an unpleasant question (\cite{Weyl}, p. 98), even unsolvable (\cite{Messiah}, II, 15) and always subtle (\cite{Penrose}, p. 497-98). This has led to a multitude of publications, many interesting from a purely mathematical point of view. But as regards the quantization of classical mechanical systems is a problem of a ghostly nature: it only appears if it is invoked.

The magnitudes of Classical Mechanics are always defined by tensor fields: the kinetic energy is (up to a constant factor) the very metric tensor, thought of as a function in the phase space; the moments are contractions of vector fields with the Liouville form, etc.

Therefore, the problem of the quantization of classical magnitudes consists, in principle, in giving a rule that canonically assigns to each tensor field in the configuration space, $M$, a differential operator on $\C^\infty(M)$.

As an example, when the configuration space is $\R^n$, its vector structure allows us to define the successive differentials, $d^rf$, of each function $f\in\C^\infty(M)$. Every $d^rf$ is a $r$-order covariant tensor field on $\R^n$.

For each symmetric contravariant tensor field $\Phi$ or order $r$, we define the differential operator $\widehat\Phi$ on $\C^\infty(\R^n)$ by the rule
\begin{equation}\label{cuantizacionplana}
\widehat\Phi(f):=(-i\hbar)^r\langle\Phi\, ,\,d^rf\rangle
\end{equation}
(tensor contraction).

Let us fix a pseudo-Riemannian metric $T_2$ with constant coefficients (= invariant under translations) in $\R^n$; $T_2$ establishes a canonical correspondence between covariant tensors and contravariant tensors (raising and lowering indexes), which allows us translate (\ref{cuantizacionplana}) to a rule of quantization for covariant tensors. For instance,
$$\widehat T_2=-\hbar^2\Delta,$$
where $\Delta$ is the Laplacian operator associated with the metric.

The assignment $\Phi\to\widehat\Phi$ is the classical ``canonical quantization'' (see, for instance, \cite{Mackey}), but presented without the ``factor-ordering problem''.

The generalization of the rule (\ref{cuantizacionplana}) to the case where $M$ is an arbitrary smooth manifold and $\nabla$ a symmetric linear connection in $TM$ is obvious: for each symmetric cotravariant tensor field $\Phi$ on $M$ of order $r$ and each $f\in\C^\infty(M)$:
\begin{equation}\label{cuantizacion}
\widehat\Phi(f):=(-i\hbar)^r\langle\Phi\, ,\,\nabla^rf\rangle
\end{equation}
(tensor contraction).

Given a pseudo-Riemannian metric $T_2$ in $M$, the quantization of covariant tensor fields is performed as above. The connection $\nabla$ can be chosen differently than Levi-Civita connection of $T_2$; the quantization of covariant tensor fields depends on the couple of data $(\nabla,T_2)$.

Note that for vector fields in $M$ (contravariant tensor fields of order 1) the quantization (\ref{cuantizacion}) does not depends on $\nabla$, because $\nabla f=df$, and so is canonically defined by the differentiable structure of $M$.

The rule of quantization (\ref{cuantizacion}) is so natural that it is highly unlikely that it is not present in the literature. But we have not found any precise reference; the closest to that we know of is \cite{LychaginQ}. % (see also \cite{Hess}).

The rule (\ref{cuantizacion}) is valid for functions on $TM$ which are polynomials along the fibers and only for these ones; no method of ``passing to the limit'' allows one to extend the quantization to another class of functions, as we will see. In this work we will show that (\ref{cuantizacion}) is equivalent, in the case of polynomial functions in fibers, to the quantization method exposed in \cite{QuantizacionMunozAlonso}; with additional assumptions on the $\nabla$ connection, the quantization rule \cite{QuantizacionMunozAlonso} allows quantization of classes of functions in $TM$ that are much broader than polynomial in fibers.

The equivalence of rules of quantization (\ref{cuantizacion}) and that given in \cite{QuantizacionMunozAlonso} is derived from a nice property of the Riemannian exponential map, $\textrm{exp}$, defined by a connection $\nabla$: $\textrm{exp}$ transforms the iterated (symmetrized) covariant differentials of each function $f\in\C^\infty(M)$ into the corresponding ``ordinary'' differentials of $\textrm{exp}^*f$ in the fibres of $TM$ (defined by means of the vector structure). Apparently, this property of the exponential map is not found in the literature.

On the other hand, the quantization method \cite{QuantizacionMunozAlonso} can be extended to a wider set of functions than just those that come from tensors (``Hamiltonian functions'', which are polynomials in fibers). To this end, we must consider those functions whose Fourier transform in fibers produces linear functionals on the ring of smooth functions and compose with the injection that determines the Riemanian exponential. These include those functions characterized by the Paley-Wiener-Schwartz theorem.

Part 1 of this work is preparatory in nature. We fix the terminology and make some observations about the classical character  of de Broglie waves (prior to any quantization). Starting from the de Broglie waves, the classical path leads us to a ``quasi Schrödinger'' equation \cite{Holland, RM, RMarxiv} by interpreting the motion of the virtual particles on each solution of the Hamilton-Jacobi equation as a fluid (see \cite{intermediatefluids} about the relations of fluids with intermediate integrals of classical mechanical systems).

By the classical way it does not seem possible to arrive at the Schrödinger equation. But we will show in Section \ref{SKG} that, in a precise sense, such an equation is the only one canonically related to Newton equations.

\bigskip

\section{Notes on Classical Mechanics and Undulatory Mechanics}\label{uno}

\subsection{Structures previous to the metric}\label{previas}

Let $M$ be a smooth manifold of dimension $n$. Let $TM$, $T^*M$ be the tangent and cotangent bundles of $M$, respectively. Let $\C^\infty(M)$ be the ring of differentiable  functions on $M$ with complex values. We will consider $\A$ as a subring of $\AT$ by means of the injection derived from the canonical projection $\pi\colon TM\to M$.

The vector fields tangent to $TM$ which (as derivations of the ring $\AT$) kill the subring $\A$, are the \emph{vertical} tangent fields. The differentiable 1-forms on $TM$ which, by interior product, kill the vertical tangent fields are the \emph{horizontal} 1-forms on $TM$. The lifting of 1-forms from $M$ to $TM$ by means of $\pi^*$ are horizontal and, locally, any horizontal 1-form on $TM$ is a linear combination of such 1-forms with coefficients in $\AT$.

Each horizontal 1-form $\alpha$ on $TM$ defines on $TM$ a function $\dot\alpha$ given by $\dot\alpha(u_x)=\langle\alpha,u_x\rangle$ (inner product), for each $u_x\in TM$. In particular, for each fuction $f\in\A$, the function $\dot{({df})}$ will be denoted, for short, $\dot f$. Essentially, $\dot f$ is $df$:
$\dot f(u_x)=\langle df,u_x\rangle=u_x(f)$ (derivative of $f$ by $u_x\in TM$). If $(x^1,\dots,x^n)$ is a system of coordinates on an open subset of $M$, the $(x^1,\dots,x^n,\dot x^1,\dots,\dot x^n)$ are coordinates on the corresponding open subset of $TM$.

Each covariant tensor field $a$ of degree $r$ on $M$ canonically defines a function $\dot a$ on $TM$, polynomial along the fibres:
$\dot a(u_x)=\langle a,\overset{\text{$r$ times}}{\overbrace{u_x\otimes\cdots\otimes u_x}}\rangle$. In local coordinates, $\dot a$ is obtained by substituting $dx^j$ by $\dot x^j$ in the expression of the tensor $a$.

The linear structure of each fibre $T_xM$, allows us to identify the tangent space to $T_xM$ at each one of its points $u_x$ with the very vector space $T_xM$: to the vector $v_x\in T_xM$ it corresponds the vector $V_{u_x}\in T_{u_x}(T_xM)$ that is the ``derivative along $v_x$''. We will say that $V_{u_x}$ is the \emph{vertical representative} of $v_x$ at $u_x$ and that $v_x$ is the \emph{geometric representative} of $V_{u_x}$.

 By going to the definitions it is checked that, for each $f\in\A$, we have $V_{u_x}(\dot f)=v_x(f)$. In this way, each tangent vector $v_x\in T_xM$ determines on its fibre $T_xM$ a tangent field which is ``constant'' (= parallel).

 Each tangent field on $M$ determines a vertical tangent field on $TM$, constant along each fibre.

As a consequence, each contravariant tensor field $\Phi$ on $M$ determines a \emph{vertical} contravariant tensor field
$\PHI$ on $TM$, constant (=parallel) along each fibre.

In local coordinates, $\PHI$ is obtained from $\Phi$ by substituting each field $\partial/\partial x^j$ by its vertical representative $\partial/\partial\dot x^j$.

A symmetric contravariant tensor field $\Phi$ of degree $r$  on $M$ determines on $T^*M$ a function $F$, polynomial on the fibres, defined by
$$F(\alpha_x)=\langle\Phi,\overset{\text{$r$ times}}{\overbrace{\alpha_x\otimes\cdots\otimes\alpha_x}}\,\rangle$$
(tensor contraction).
We will say that $F$ is the \emph{Hamiltonian associated} with $\Phi$.

If $(x^1,\dots,x^n)$ are coordinates on an open subset of $M$, the Hamiltonian associated with $\partial/\partial x^j$ is usually denoted by $p_j$:
$$p_j(\alpha_x)=\langle\alpha_x,\partial/\partial x^j\rangle.$$
The functions $(x^1,\dots,x^n,p_1,\dots,p_n)$ are local coordinates on $T^*M$. For a given contravariant tensor field $\Phi$ on $M$, its associated Hamiltonian is obtained by substituting in the expression of $\Phi$ each $\partial/\partial x^j$ by $p_j$.

Symmetric contravariant tensor fields on $M$, homogeneous or not, canonically corresponds to the functions $\in\C^\infty(T^*M)$ that are polynomials along the fibres. We will refer to this particular type of functions as \emph{Hamiltonians}.

 In  $T^*M$ it is defined the \emph{Liouville 1-form} $\theta$ by
 $\theta_{\alpha_x}=\pi^*\alpha_x$, for each $\alpha_x\in T^*M$ ($\pi\colon T^*M\to M$ is the canonical projection). We will simplify the notation by putting $\theta_{\alpha_x}=\alpha_x$, understanding that covariant tensors in general rise from $M$ to $T^*M$ by ``pull-back'' through $\pi^*$.
 In local coordinates $(x^1,\dots,x^n,p_1,\dots,p_n)$, we have $\theta=p_j dx^j$.

 The 2-form $\omega_2:=d\theta$ is the \emph{symplectic form} on $T^*M$. In local coordinates, $\omega_2=dp_j\wedge dx^j$.

 The 2-form $\omega_2$ has no kernel, so establishes an isomorphism between the $\C^\infty(T^*M)$-module of tangent fields on $T^*M$ and that of the 1-forms on $T^*M$:
 $$D\mapsto \alpha:=D\,\lrcorner\,\omega_2.$$

 The structure of Lie algebra (given by the commutator) in the module of tangent fields is translated to the module of 1-forms defining an structure of Lie algebra given by the \emph{Poisson bracket}. The Poisson bracket of two closed 1-forms is an exact 1-form. However, in order not to leave arbitrary constants, a Poisson bracket of functions must be defined: for each function $F\in\C^\infty(T^*M)$ the  \emph{Hamiltonian field} of $F$ is defined by the condition $D_F\,\lrcorner\,\omega_2=dF$.

 The Poisson bracket of two functions $F$, $G$, is defined by
 $$\{F,G\}:=D_FG$$
 (which equals $-D_GF=-2\omega_2(D_F,D_G)$).

  For the local coordinates $(x^j,p_j)$, the Hamiltonian fields are $D_{x^j}=\partial/\partial p_j$, $D_{p_j}=-\partial/\partial x^j$, so that
  $$\{x^j,x^k\}=0,\quad \{p_j,p_k\}=0,\quad\{p_j,x^k\}=-\delta_j^k.$$

 In order to avoid confusions with the terminology, let us observe that the Hamiltonian function associated with the tensor $\partial/\partial x^j$ is $p_j$, while $\partial/\partial x^j$ is the Hamiltonian field of function $-p_j$.

We have seen that to each covariant tensor field $a$ of order $r$ on $M$ it corresponds a function $\dot a$ on $TM$ polynomial on the fibres. On $T^*M$, the symplectic structure $\omega_2$ makes a vertical tangent field $\widetilde\alpha$ on $TM$ correspond to each 1-form $\alpha$ on $M$, by the  rule $\widetilde\alpha\,\lrcorner\,\omega_2=\alpha$. In local coordinates, the field $\widetilde\alpha$ which corresponds to $dx^j$ is $\partial/\partial p_j$.
For arbitrary order $r$, the correspondence established by the symplectic structure assigns to each symmetric covariant tensor field $a$ of order $r$, a symmetric contravariant tensor field $\PHI_a$ of order $r$ and ``vertical'' (its contraction with any ``horizontal'' tensor vanishes); in local coordinates, $\PHI_a$ is obtained by substituting each $dx^j$ by $\partial/\partial p_j$ in the expression of $a$. This tensor field $\PHI_a$ gives on each fibre of $T^*M$ a differential operator of order $r$ that does not depend on the coordinates $(x^1,\dots,x^n)$, because changes of local coordinates on $M$ give always linear changes of coordinates in the fibres of $T^*M$ (and also of $TM$). By acting fiberwise it is obtained a differential operator $\widetilde\PHI_a$, on $\C^\infty(T^*M)$, that kills the subring $\A$.

Therefore, we have the correspondences
$$\dot a(x,\dot x)\quad\longleftrightarrow\quad a(x,dx)\quad\longleftrightarrow\quad \PHI_a\quad\longleftrightarrow\quad\widetilde\PHI_a(x,\partial/\partial p).$$
$\PHI_a$ is the polynomial in the $\partial/\partial p_j$ that results by the substitution in the expression of $a$ each $dx^j$ by $\partial/\partial p_j$.

The correspondence $\dot a\to\widetilde\PHI_a$ is modified by a constant factor $\dot a\to k^{-r}\widetilde\PHI_a$ (for tensors of order $r$) if the symplectic form $\omega_2$ is changed to $k\omega_2$ (where $k\in\mathbb{C}$ is arbitrary). In Quantum Mechanics, it is taken $k=i/\hbar$.

The same association $$\text{\{covariant symmetric tensor on $M$\} $\to$ \{vertical differential operator on $TM$\}}$$ is obtained from the Fourier transform, by using the linear duality between $TM$ and $T^*M$. Let us denote by $\mathcal{S}(TM)$ the space of complex functions on $TM$ which, when restricted to each fibre $T_xM$, are of class $\C^\infty$ and rapidly decreasing they and all of their derivatives. Analogous meaning for $\mathcal{S}(T^*M)$. On each fibre $T_xM$ (being a $\R$-linear space) there is a measure invariant by translation $\mu$, univocally determined up to a multiplicative constant (``Haar measure''). Once fixed that factor for each $T_xM$, it is defined the Fourier transform  $\mathcal{S}(TM)\to\mathcal{S}(T^*M)$ by
$$(\mathcal{F}f)(\alpha_x):=\int_{T_xM} f(v_x)e^{\frac i\hbar\langle v_x,\alpha_x\rangle}\,d\mu(v_x).$$
In local coordinates:
$$(\mathcal{F}f)(x,p):=\int_{T_xM} f(x,\dot x)e^{\frac i\hbar p_j\dot x^j}\,d\dot x^1\cdots d\dot x^n.$$
(the constant factor that affects the integral is irrelevant for the following).

By differentiation under the integral sign we get the classical formula
$$\F(\dot a(x,\dot x)f(x,\dot x))=a(x,-i\hbar\partial/\partial p)(\F f)(x,p),$$
that is, $\F\circ\dot a=\widehat\PHI_a\circ\F$, where $\widehat\PHI_a$ is the vertical differential operator which results of substituting in the  tensor $a$ each $dx^j$ by $-i\hbar\partial/\partial p_j$.

This is the correspondence given by the symplectic structure $(i/\hbar)\omega_2$·

For later references, let us write the correspondence between symmetric covariant tensor fields on $M$ and vertical differential operators on $T^*M$, once $\omega_2$ is substituted by $(i/\hbar)\omega_2$:
\begin{equation}\label{corr1}
\dot a=a(x,\dot x)\quad\longleftrightarrow\quad a=a(x,dx)\quad\longleftrightarrow\quad \widehat\PHI_a(x,\partial/\partial p)=a(x,-i\hbar \partial/\partial p).
\end{equation}
\bigskip

\subsection{Introduction of a metric. Classical mechanical systems}\label{intrometrica}

 Let $T_2$ be a pseudo-Riemannian metric (non degenerate of arbitrary signature) on the manifold $M$. Such a metric determines an isomorphism of fibre bundles $TM\simeq T^*M$, that allows us to transport from one to each other all the structures that we have considered. Hence, the Liouville form $\theta$ and the symplectic form $\omega_2$ passe from $T^*M$ to $TM$, where we will denote them in the same way.

  If the expression of the metric in local coordinates is $T_2=g_{jk}(x)dx^j\,dx^k$, the isomorphism $TM\simeq T^*M$ is expressed by the equations
  $p_j=g_{jk}\dot x^k$. The differential operators $\partial /\partial p_j$ transported to $TM$ become $g^{jk}\partial/\partial \dot x^k$, and in the correspondence (\ref{corr1}), the operator $\widehat\PHI_a$ is $a(x,-i\hbar g^{jk}\partial/\partial\dot x^k)$. For the 1-form $\alpha_j=g_{jk} dx^k$, it holds $\dot\alpha_j=p_j$, and the corresponding operator $\widehat\PHI_a$ is $g_{jk}(-i\hbar g^{k\ell}\partial/\partial \dot x^\ell)=-i\hbar\partial/\partial\dot x^j$:
  \begin{equation}\label{corr2}
  p_j\mapsto -i\hbar\frac{\partial}{\partial\dot x^j}
  \end{equation}
  in the correspondence of functions on $TM$ linear along the fibres with vertical tangent fields.

  In coordinates of $TM$, $\theta=g_{jk}\dot x^jdx^k$ and the function associated with $\theta$ on $TM$ is
  $\dot\theta=g_{jk}\dot x^j\dot x^k=2T$ where $T$ is the \emph{kinetic energy} function.

  On $TM$ the Hamiltonian tangent field for the function $-T$ is the \emph{geodesic field} of $(M,T_2)$; according its very definition, it holds
  \begin{equation}\label{geodesico1}
  {D_G}\lrcorner\,\omega_2+dT=0
  \end{equation}

  For later references, the well known expression of the geodesic field is
  \begin{equation}\label{geodesico2}
  D_G=\dot x^j\frac\partial{\partial x^j}-\Gamma_{k\ell}^j(x)\dot x^k\dot x^\ell\frac\partial{\partial \dot x^j}
  \end{equation}
  where the $\Gamma$' are the Christoffel symbols of the metric (within our convention,
   $\Gamma_{k\ell}^j=\left\{{\substack{j\\ k\ell}}\right\}$).

  Let us recall that a \emph{second order differential equation} on $M$ is, by definition, a tangent field $D$ on $TM$ such that, as a derivation, takes each $f\in\A$ to $Df=\dot f$. Thereby, $D_G$ in (\ref{geodesico2}) is a second order differential equation.

  Two second order differential equations on $M$ derive in the same way the subring $\A$ of $\AT$. Thus, any second order differential equation $D$ on $M$ is of the form $D=D_G+V$, where $V$ is a vertical tangent field on $TM$.

  The vertical tangent fields are the \emph{forces} of the Classical Mechanics. A classical-mechanical system  is a set comprised by three data $(M,T_2,V)$ and the \emph{Newton law} says that the evolution of the space of states $TM$ is the flow of the field $D=D_G+V$. In particular, when $V=0$, the system evolves according the geodesic flow (\emph{inertial law}).

  In order to ``visualize'' a force $V$ in an state $u_x\in TM$ (``position-velocity state'') we must translate the vertical vector $V_{u_x}$ to its geometrical representative $v_x$. Once this is done, the Newton law can be stated in the original form ``force = mass $\times$ acceleration'': the trajectory of the field $D$ that passes through the point $u_x$ is projected onto $M$ as a curve whose tangent field $u$ (defined along the curve) holds $v_x=\nabla_{u_x}u$. The left member is the ``force'' and the right member is  the ``mass times acceleration'', understood that masses and inertial moments are incorporated as factors in $T_2$. For further details see (\cite{MecanicaMunoz}, Section 1).

  By looking at the coordinate expression of the symplectic form, we immediately see that, in the correspondence established by $\omega_2$ between tangent fields and 1-forms on $T^*M$, the vertical fields correspond exactly with horizontal 1-forms: $i_V\omega_2=-\alpha$, horizontal. By applying this equality to the field $D=D_G+V$ it results
  \be\label{Newton}
  D\lrcorner\,\omega_2+dT+\alpha=0.
  \ee
  Equation \ref{Newton} expresses the biunivocal correspondence between second order differential equations on $M$ and horizontal 1-forms on $TM$. The form $\alpha$ is the \emph{work form} of the mechanical system.

  A  mechanical system $(M,T_2,\alpha)$ is said to be \emph{conservative} when the $\alpha$ is an exact differential form, $dU$. By taking into account that $\alpha$ is horizontal, $U$ have to be a function $\in\A$. The sum  $H=T+U$ is the \emph{Hamiltonian} of the system, and (\ref{Newton}) is
  \begin{equation}\label{hamiltoniano1}
  D\lrcorner\,\omega_2+dH=0
  \end{equation}

   $D$ is the Hamiltonian field of the function $-H$ in the terminology of Section \ref{previas}.  Equation (\ref{hamiltoniano1}), when is written in coordinates of $T^*M$, is the system of Hamilton canonical equations.

   Let us highlight the following consequence of (\ref{Newton}) that will be necessary in Section \ref{SKG}:

   \begin{prop}
   The Hamiltonian fields on $T^*M$ that, by means of the metric $T_2$, are transferred to $TM$ as second order differential equations are exactly those that govern the evolution of conservative mechanical systems on $(M,T_2)$ through (\ref{hamiltoniano1}).
   \end{prop}

   No other infinitesimal canonical transformation on $T^*M$ is the law of evolution of a mechanical system on $ (M, T_2) $.

 A tangent field $u$ in $M$ is an \emph{intermediate integral} of the field $D$ when the solution-curves of $u$ in $M$, lifted as curves to $TM$ (each point $x$ of the curve goes to the point $(x,u_x)$ of $TM$) is also a solution of $D$. We can think about a given vector field $u$ as a section of the fibre bundle $TM\to M$; when passing to $T^*M$, the section $u$ corresponds to a section $\alpha$ of $T^*M\to M$ where $\alpha=u\lrcorner\, T_2$ or, in other words, $u=\textrm{grad}\,\alpha$. If the section $\alpha$ is a Lagrangian submanifold of $T^*M$, locally $\alpha=dS$ for a certain function $S$ on $M$ (or on some open subset). In such a way, the necessary and sufficient condition for $u=\textrm{grad}\,S$ to be an intermediate integral of the Hamiltonian field $D$ in (\ref{hamiltoniano1}) is that
 \begin{equation}\label{HJ}
 H(\textrm{grad}\,S)=E,\quad\text{constant,}
 \end{equation}
 where $H(\textrm{grad}\,S)$ is the specialization oh $H\in\AT$ to the section $u=\textrm{grad}\,S$. (\ref{HJ}) is the \emph{Hamilton-Jacobi equation} (see \cite{RM,RMarxiv}).

\subsection{De Broglie waves and Schrödinger equation}

Let us consider a conservative mechanical system with configuration space $(M,T_2)$ and Hamiltonian $H=T+U$. Let $S\in\C^\infty(M)$ be a solution of the Hamilton-Jacobi equation (\ref{HJ}); $\textrm{grad}\,S$ is an intermediate integral of the equations of motion. Let $\textrm{Grad}\,S$ be the vertical field on $TM$ whose geometric representative is $\textrm{grad}\,S$; we have $\textrm{Grad}\,S\lrcorner\,\omega_2=\textrm{grad}\,S\lrcorner\, T_2=dS$, whereby $\textrm{Grad}\,S$ is the Hamiltonian field whose Hamiltonian function is $S$.

The correspondence (\ref{corr1}) applied to the tensor $dS$ is
$$\dot S\quad\leftrightarrow \quad dS=\frac{\partial S}{\partial x^j}\,dx^j\quad\leftrightarrow\quad -i\hbar\frac{\partial S}{\partial x^j}\frac{\partial }{\partial p_j}=     -i\hbar\,\textrm{Grad}\,S.$$

 For first order differential operators we have a rule, \emph{previous to any quantization rule}, which assigns a field on $M$ to each vertical vector field constant along the fibres of $TM$: to go from a vertical field to its geometric representative. In this case, $-i\hbar\,\textrm{Grad}\,S\mapsto-i\hbar\,\textrm{grad}\,S$.

 The correspondence:
$$\text{\emph{Classical magnitude (function on $TM$)}, $\dot S$}\,\,\to\,\,
  \text{\emph{Differential operator on $\C^\infty(M)$}, $-i\hbar\,\textrm{grad}\,S$}$$
  must remain valid in any quantization law.

  On the section $\textrm{grad}\,S$ of $TM$, the function $\dot S$ takes the value $\dot S\mid_{\textrm{grad}\,S}=\langle dS,\textrm{grad}\,S\rangle=\|\textrm{grad S}\|^2=2(E-U)$.
  The functions $\varphi$ on $M$ on which the classical magnitude $\dot S\mid_{\textrm{grad}\,S}$ and its associated differential operator, $-i\hbar\,\textrm{grad}\,S$, act (the first one by means of multiplying) giving the same result, are those which hold the differential equation:
  \begin{equation}\label{Broglie1}
  -i\hbar\,\textrm{grad}\,S(\varphi)=2(E-U)\cdot\varphi.
  \end{equation}

  The parameter $t$ proper for the trajectories of the vector field $\textrm{grad}S$ holds on each trajectory
  \be\label{tyS}
  dt=\frac{dS}{\|\textrm{grad}\,S\|^2}=\frac{dS}{2(E-U)}.
  \ee

  By changing the parameter $t$ by the parameter $S$  on each trajectory, Equation (\ref{Broglie1}) is
  $$-i\hbar\frac{d\varphi}{dS}=\varphi,$$
  which gives
  \begin{equation}\label{Broglie2}
  \varphi=\varphi_0e^{i\frac S\hbar}=\varphi_0e^{2\pi i\frac Eh\cdot\frac SE},
  \end{equation}
  where $\varphi_0$ is an arbitrary first integral of the vector field $\textrm{grad}\,S$.

  The wave function $\varphi$ is derived form the condition of that, on it, give the same result the action  of the differential operator $-i\hbar\,\textrm{grad}\,S$ and (by multiplication) the classical magnitude $\dot S$ (from which that operator proceeds) restricted to the section $\textrm{grad}\,S$. The generalization of this principle of formation of wave equations is that used in \cite{QuantizacionMunozAlonso}.

  Note that the advance rate of the wavefronts for $\varphi$ in (\ref{Broglie2}) is uniform if the quotient  $S/E$ is used as time, while the time $t$ that measures the motion of the virtual particles in the mechanical system holds  (\ref{tyS}). These two times are different, except for the geodesic field ($U=0$). In general, it seems rather that time in Quantum Mechanics is a linear combination of both times, $t$ and $S/E$, as shown in detail in \cite{TiempoMunozAlonso}.

Going back to (\ref{Broglie2}) and taking constant $\varphi_0$ an straightforward computation gives
the identity
 \begin{equation}\label{Broglie4}
  \left(-\frac{\hbar^2} 2\Delta+U\right)\varphi=\left(\frac{\hbar}{2i}\Delta S+H(\textrm{grad} S)\right)\varphi.
  \end{equation}
  From that identity it is derived the
  \begin{prop}[\cite{RM}]
  Let $(M,T_2,dU)$ be a conservative mechanical system. Let $S\in\C^\infty(M)$. From the three conditions
  \begin{enumerate}
  \item[A)] $S$ holds the Hamilton-Jacobi equation (\ref{HJ})
  \item[B)] $S$ is harmonic: $\Delta S=0$
  \item[C)] $\varphi=e^{iS/\hbar}$ holds the Schrödinger equation $\left(-\frac\hbar 2\Delta+U\right)\varphi=E\varphi$
  \end{enumerate}
  each couple of them implies the third one.

  When $M$ is oriented, the metric $T_2$ gives a volume form and condition B) can restated as
  \begin{enumerate}
  \item[B')] $-i\hbar\,\textrm{grad}\,S$ is self-adjoint.
  \end{enumerate}
  \end{prop}
  \bigskip

The phase  $S$ determines the energy $E=H(dS)$ (the ``frequency'') and the velocity of propagation $u=\textrm{grad}\, S$ (up to a factor). The  intensity, wave ``energy'', amplitude, ..., are a different thing to the $E$. As with sound waves, the intensity does not change the frequency or velocity of propagation. Any way we try to interpret the amplitude of a wave (energy density, probability of presence, etc.) must be related to something that is preserved by the flow $u$. When we imposed the condition $\textrm{div} \, u=0$ (conservation of volume), we obtained as a condition the Schrödinger equation for constant amplitude wave functions. Let us impose the condition that $u$ preserve a $n$-form $\rho \, \omega_n$ ($\omega_n$ is the volume form); that condition is $u(\rho)+\rho\,\textrm{div}\,u=0$, that is
\[
\begin{split}
T_2(\textrm{grad}\, S,\textrm{grad}\,\rho)&+\rho\,\Delta S=0,\quad\text{or},\\
u(\textrm{log}\,\rho)&+\Delta S=0\quad\text{(\emph{conservation law}).}
\end{split}
\]
This equation shows that the possible amplitudes in the Lagrangian manifold $u=\textrm{grad}\,S$ are determined up to a factor, which is a first integral of the flow $u$.

\subsection*{Waves on $u=\textrm{grad}\,S$}
Chosen the function $\rho$ so that the flow $u$ preserves the density $\rho\,\omega_n$, let us study the equation satisfied by the wave function
$$\Phi=f(\rho)\,e^{iS}$$
where $f$ is a function, at the moment undetermined. Making the usual calculations, we obtain the identity:
\begin{multline*}
 \phantom{mm}\Delta\Phi=\left\{\frac{f'(\rho)}{f(\rho)}\,\Delta\rho+\frac{f''(\rho)}{f(\rho)}\|\textrm{grad}\,\rho\|^2-\|\textrm{grad}\,S\|^2+\right.\\
     +\left.i\left[\Delta S+2\,\frac{f'(\rho)}{f(\rho)}\,T_2(\textrm{grad}\,S,\textrm{grad}\,\rho)\right]\right\}\, \Phi.\phantom{mm}
\end{multline*}
By adding the conservation law, the imaginary part of $\{\phantom{i}\}$ is
$$\Delta S\left(1-\frac{2\rho f'(\rho)}{f(\rho)}\right).$$
Leaving aside the hypothesis that $S$ is harmonic, the cancellation of that imaginary part requires that, except for one factor, be
 $f(\rho)=\sqrt{\rho}$ or $\rho=|\Phi|^2$, which is consistent with the usual interpretations of $\rho$.

For such a  $f$, we get
\begin{equation*}
\left[\Delta+\|\textrm{grad}\,S\|^2\right]\Phi=\left[\frac{\Delta \rho}{2\rho}-\frac 1{4\rho^2}\|\textrm{grad}\,\rho\|^2\right]\Phi
  =\frac{\Delta\sqrt{\rho}}{\sqrt{\rho}}\,\Phi,
\end{equation*}
so that
$$\left(-\frac 12\,\Delta+U\right)\,\Phi=\left(H-\frac{\Delta\sqrt{\rho}}{2\sqrt{\rho}}\right)\,\Phi,$$
from which we derive the following:
\begin{prop}
Let $(M,T_2,dU)$ be a conservative system, $S$ a function on $M$, $u=\textrm{grad}\,S$, $\rho$ a function such that $\mathcal{L}_u(\rho\,\omega_n)=0$. Then of the three conditions:

\begin{enumerate}[A)]
\item  $S$ holds the Hamilton-Jacobi equation
 $\displaystyle{H(dS)=E}$ \vskip .2cm
\item $\sqrt{\rho}$ is harmonic
\item $\Phi=\sqrt{\rho}\,e^{iS}$ holds the Schrödinger equation
$\displaystyle{\left(-\frac 12\Delta+U\right)\,\Phi=E\,\Phi}$
\end{enumerate}
each pair implies the third.
\end{prop}

In relation to this section, compare with points 2.6, 2.7 in Holland \cite{Holland}.

\bigskip

\section{Quantization of contravariant tensors. Dequantization of differential operators}\label{tres}

Let $M$ be an smooth manifold, $\C^\infty(M)$ its ring of $\C^\infty$ functions with complex values. The tensor fields that we consider have complex valued functions as coefficients. Let $\nabla$ be a symmetric linear connection on $TM$.

\begin{defi}[Quantization defined by $\nabla$]
 For each symmetric contravariant tensor field of order $r$, $\Phi$, on $(M,\nabla)$, the \emph{quantized of $\Phi$} by $\nabla$ is the differential operator $\widehat\Phi$ which, for each  $f\in\A$ gives
\be\label{quantizado}
\widehat\Phi(f):=(-i\hbar)^r\langle\Phi,\nabla^r_{\textrm{sym}}f\rangle
\ee
where $\langle\,,\,\rangle$ denotes tensor contraction and $\nabla^r_{\textrm{sym}}f$ is the symmetrized tensor of the $r$-th covariant iterated differential of $f$ with respect to the connection $\nabla$.

The quantized of a non-homogeneous tensor is the sum of the quantized of its homogeneous components.
\end{defi}

\begin{obs}
This definition can be generalized giving a differential operator between sections of fibre bundles for each contravariant tensor $\Phi$ on $M$, once a linear connection is fixed in the first fibre bundle. This generalization will not be considered in what follows.
\end{obs}

Let us recall that a differential operator of order $r$ on $M$ (= differential operator of order $r$ on $\A$) is an $\Com$-linear map $P\colon\A\to\A$ which holds the following condition: for each point $x\in M$, $P$ takes the ideal $\m_x^{r+1}$ into $\m_x$ ($\m_x$ is the ideal of the functions of $\A$ vanishing at $x$).

It is derived that $P$ takes the quotient $\m_x^{r}/\m_x^{r+1}$ into $\A/\m_x=\Com$. By taking into account that $\m_x^{r}/\m_x^{r+1}$ is the space of symmetric covariant tensors of order $r$ at the point $x$ (homogeneous polynomials of degree $r$, with coefficients in $\Com$, in the $d_xx^1,\dots,d_xx^n$, once taken local coordinates), we see that $P$ determines a symmetric contravariant tensor of order $r$
 called \emph{symbol of order $r$ of $P$ at $x$}, denoted by $\sigma_x^r(P)$,
\be\label{simbolo}
\sigma_x^r(P)\colon\m_x^{r}/\m_x^{r+1}=T^{*r}_xM\to\R,
\ee
that is the map canonically associated with $P$ by pass to the quotient.

When $x$ runs over $M$, we get the tensor field $\sigma^r(P)$ on $M$ called \emph{symbol of order $r$ of $P$}. If $\sigma^r(P)=0$, $P$ is of order $r-1$.

In the case $M=\R^n$, with vector coordinates $x^1,\dots,x^n$, let us denote $\partial^\alpha$ the tensor
$\partial^\alpha:=(\partial/\partial x^1)^{\alpha_1}\cdots(\partial/\partial x^n)^{\alpha_n}$. Its quantized by the rule (\ref{quantizado})  (with the vector connection of $\R^n$) is $\widehat{\partial^\alpha}:=(-i\hbar)^{|\alpha|}D^\alpha$, where $D^\alpha$ is the differential operator ${\partial^{|\alpha|}}/{(\partial x^1)^{\alpha_1}\dots(\partial x^n)^{\alpha_n}}$. It is directly seen that $\sigma^{|\alpha|}(D^\alpha)=\partial^\alpha$, so that for any tensor field of order $r=|\alpha|$ on $\R^n$ is obtained, by adding terms,
\be\label{simboloplano}
\sigma^r\left(\widehat\Phi\right)=(-i\hbar)^r\Phi
\ee

Going from $\R^n$ to the general case $(M,\nabla)$ let us observe that, when the iterated covariant differentials of a function $f$ are calculated in local coordinates, the derivatives of maximum order, $r$, of $f$ appear in terms which does not contain Christoffel  symbols (as in the case of $\R^n$). Since the symbol of an operator of order $r$ depends only on these terms, Formula (\ref{simboloplano}) is still valid in general for the quantization rule
(\ref{quantizado}) on $(M,\nabla)$.

\begin{thm}\label{tquantizado1}
The rule of quantization (\ref{quantizado}) establishes a biunivocal correspondence between linear differential operators $P$ and symmetric contravariant tensor fields (not necessarily homogeneous) on $M$. To the operator $P$ of order $r$ corresponds the tensor $\Phi=\Phi_{r}+\Phi_{r-1}+\cdots\Phi_{0}$ (each $\Phi_j$ denotes the homogeneous component of degree $j$) such that
$$
\sigma^r(P)=(-i\hbar)^r\Phi_r
$$
and, for $k=1,\dots,r$:
$$\sigma^{r-k}(P-\widehat\Phi_{r}-\cdots-\widehat\Phi_{r-k+1})=(-i\hbar)^{r-k}\Phi_{r-k}$$
and
\be\label{quantizado3}
P=\widehat\Phi=\widehat\Phi_{r}+\widehat\Phi_{r-1}+\cdots+\widehat\Phi_{0}.
\ee
\end{thm}

\begin{defi}[Dequantization]\label{dquantizado2}
The contravariant tensor $\Phi$ in (\ref{quantizado3}) is the \emph{dequantized of the differential operator $P$ by the connection $\nabla$}.
\end{defi}

We have seen in Section \ref{previas} that symmetric contravariant tensor fields (homogeneous or not) on $M$ canonically correspond with functions $F\in \AC$ polynomials along the fibres.

\begin{defi}\label{dquantizado3}
The function $F\in\AC$ corresponding to the tensor $\Phi$ dequantized of the differential operator $P$ will be called \emph{Hamiltonian of $P$} with respect to the connection $\nabla$.
\end{defi}

The symplectic structure $\omega_2$ of $T^*M$ assigns to each $F\in\AC$  a Hamiltonian vector field $D_F$, as  we have already remembered in section \ref{previas}, by the rule $D_F\lrcorner\,\omega_2=dF$.
A complex vector field in $M$ (= a derivation of the ring $ \C^\infty(M) $) is separated into real and imaginary parts as a pair of real tangent fields in $ M $. This way, in general, $ D_F $ is considered a pair of real tangent fields, each of which is an \emph{infinitesimal canonical transformation} (i.c.t.) in the sense of Lie \cite{Lie1,Lie2}. The i.c.t. are  infinitesimal generators of (local) uniparametric groups of automorphisms of the manifold $ T^*M $ that preserve its symplectic structure.

\begin{defi}[Hamiltonian field associated with a differential operator] We will call \emph{infinitesimal canonical transformation associated with the differential operator $P$ or Hamiltonian field associated with $P$} to the tangent field $D_P$ on $T^*M$ such that
$${D_P}\lrcorner\,\omega_2+dF=0,$$
where $F$ is the Hamiltonian of $P$. (Note that, in general, $F$ is a complex function, so that $D_F$ is a couple of i.c.t. in the sense of Lie).
\end{defi}

The path $P\to F\to D_P$ is univocal. The reverse path $D_P\to F$ determines $F$ up to a additive constant; then, $F\to P$ is univocal. Thus, up to an additive constant for $P$, the correspondence $P\leftrightarrow D_P$ is biunivocal.
\begin{thm}\label{tquantizado2}
The symmetric linear connection $\nabla$ on $M$ canonically establishes a biunivocal correspondence between linear differential operators $P$, quantized of real tensors, on $\A$ (up to additive constants) and infinitesimal canonical transformations of the simplectic manifold $T^*M$ corresponding to functions polynomial along fibres (Hamiltonians).
\end{thm}

\bigskip
\noindent\textbf{Remarks on the quantization rule (\ref{quantizado}).}

\smallskip
\noindent (1) Let $\Phi_r$ be the homogeneous tensor of order $r$ that corresponds to $P$ by (\ref{quantizado3}). Considered as a function on $T^*M$, $\Phi_r$ is the function $F_r$, homogeneous of degree $r$ on the fibres. Let us suppose that the function $F_r$ is real. The first order partial differential equation $F_r((dS))=\langle\Phi_r,dS^r\rangle=0$ has as solutions the functions $S$ such that in the hypersurfaces $S=\text{const}$ the initial conditions problem for the differential operator $P$ cannot be treated by the Cauchy-Kowalevski method; they are the  \emph{characteristic hypersurfaces} of $P$ (for instance, for $\Delta$, the equation of characteristics is $\|dS\|^2=0$, the ``Eikonal equation''). The Hamiltonian field of $F_r$ has as solutions the \emph{bicharacteristics} of $P$. This field does not coincide, in general, with $D_P$. The field which propagates the singularities of $P$ is the Hamiltonian field of $F_r$, not the one of the total Hamiltonian of $P$, $D_P$.

\smallskip
\noindent (2) The relationship between the Poisson bracket of two Hamiltonians and the commutator of the corresponding quantum operators is:
\be\label{conmutadosimbolos}
\sigma^{r+s-1}[\widehat\Phi,\widehat\Psi]=-\{\sigma^r(\widehat\Phi),\sigma^s(\widehat\Psi)\}
\ee
where $\Phi$ is a tensor of order $r$, $\Psi$ of order $s$, $[\,,\,]$ is the commutator of quantized tensors and $\{\,,\,\}$ is the Poisson bracket of $\sigma^r(\widehat\Phi)$, $\sigma^s(\widehat\Psi)$, by identified with the functions they define on $T^*M$ (the minus sign proceeds from the convention taken in \ref{previas} for the Poisson bracket).

 Formula (\ref{conmutadosimbolos}) is valid for every connection $\nabla$, and its proof can be done as in the case of $\R^n$ with the vector connection, since only highest order terms of the operator intervene and we can use the symplectic form $\omega_2=dp_j\wedge dx^j$ as in vector coordinates. The checking of (\ref{conmutadosimbolos}) is a simple calculation.

 The analogy with Classical Mechanics that lead Dirac \cite{Dirac} to take the commutator as the ``Poisson bracket'' of quantistic operators, is expressed in (\ref{conmutadosimbolos}), but not in the reverse relation $\widehat{\{\Phi,\Psi\}}\overset{?}=-[\widehat\Phi,\widehat\Psi]$ which, as it is obvious and well known, is false in general, even though it is valid for first order operators.

\smallskip
\noindent (3) Differential operators representing physical magnitudes have to be selfadjoint with respect to the metric $T_2$ given in the configuration space. It is not easy to give necessary and sufficient conditions to assure that a given tensor field $\Phi$ is selfadjoint, starting from the Levi-Civita connection determined by $T_2$. We will show later that (\ref{quantizado}) is equivalent to the rule of quantization given in \cite{QuantizacionMunozAlonso}; there, we did the calculations for operator of orders $\le 3$; in order 3, computations are already very heavy. It is obtained that, for a real tensor $\Phi$,
\begin{equation*}\widehat\Phi\quad\text{formally selfadjoint}\quad\Longleftrightarrow\quad \textrm{div}\,\Phi=0.\end{equation*}
Thus, it is reasonable to conjecture that this condition holds for arbitrary order $r$.

\bigskip

\section{On the Schrödinger and Klein-Gordon equations}\label{SKG}
% \vskip 3cm

In Section \ref{intrometrica} we have recalled that a classical mechanical system is defined by three data $(M,T_2,\alpha)$: an smooth manifold $M$ (the configuration space), a pseudo-Riemannian metric $T_2$ on $M$ (in whose coefficients are incorporated the dynamical data of that system: masses, etc.) and a horizontal 1-form $\alpha$ (1-form of work). These data determine, by means of equation (\ref{Newton}) (the \emph{Newton law}), a tangent field $D$ in $TM$ (whose flow is the law of evolution of the system) which is a second order differential equation. The metric $T_2$ establishes an isomorphism between $TM$ and $T^*M$, that allows transporting their structure from one to the other bundle. By means of this transportation it has a sense to say if a field on $T^*M$ is a second order differential equation or if a tangent field on $TM$ is a infinitesimal canonical transformation. The second order differential equations  which are, at the same time, infinitesimal canonical transformations are those fields $D$ on $TM$ which (locally) hold Equation (\ref{hamiltoniano1}).

Let $\nabla$ be the Levi-Civita connection associated with $T_2$; $\nabla$ determines a rule of quantization and dequantization that puts in univocal correspondence differential operators on $\C^\infty(M)$ and contravariant tensor fields in $M$, that is to say, functions on $T^*M$ which are polynomial along the fibres (Theorem \ref{tquantizado1} and Definitions \ref{dquantizado2} and \ref{dquantizado3}). Moving on, the symplectic form $\omega_2$ assigns to each function on $T^*M$ a Hamiltonian field. In this way, it is established the correspondence of Theorem \ref{tquantizado2}).

The following theorem shows that the only differential operators that correspond, in this sense, to tangent fields that govern the evolution of classical mechanical systems, are those of the Schrödinger type:

\begin{thm}\label{teorema3}
Let $T_2$ be a pseudo-Riemannian metric in $M$, and $\nabla$ be the Levi-Civita connection of $T_2$ in $TM$. The necessary and sufficient condition for a real differential operator $P$ on $\C^\infty(M)$ to have, as its associated infinitesimal transformation, $D_P$, a second order differential equation is that $P$ be of the form
\[ P=-\frac{\hbar^2}2\Delta+U
\]
where $\Delta$ is the Laplacian operator of the metric and $U\in\C^\infty(M)$ (real).
\end{thm}
\begin{proof}
Let us start by checking that the tensor $\Phi$, the contravariant form of the metric tensor, has $\widehat\Phi=-\hbar^2\Delta$ as its quantized operator. In local coordinates, with $T_2=g_{jk}\,dx^jdx^k$, is
\[\Phi=g^{rs}\,\frac{\partial}{\partial x^r}\otimes \frac{\partial}{\partial x^s}.
\]
The expression of the second iterate covariant differential is
\[\nabla^2f=\left(\frac{\partial^2f}{\partial x^j\partial x^k}-\Gamma_{jk}^\ell\frac{\partial f}{\partial x^\ell}\right)dx^j\otimes dx^k;
\]
by contracting with $\Phi$ we get:
\[\langle\Phi\, ,\,\nabla^2f\rangle=g^{jk}
              \left(\frac{\partial^2f}{\partial x^j\partial x^k}-\Gamma_{jk}^\ell\frac{\partial f}{\partial x^\ell}\right)=\Delta f.
\]
Incorporating the factor $(-i\hbar)^2$ into $\Phi$ we see that the quantization of $\Phi$ is $-\hbar^2\Delta$.

The Hamiltonian function corresponding with the tensor $\Phi$ is $g^{rs}p_rp_s=2T$ (where $T$ is the kinetic energy function). Finally, for the Hamiltonian $H=T+U$, the corresponding quantum operator is $(-\hbar^2/2)\Delta+U$.

 Dequantizing, we pass from the operator $(-\hbar^2/2)\Delta+U$ to the Hamiltonian $T+U=H$, and then to the Hamiltonian field $D_P$ such that
 $D_P\lrcorner\,\omega_2+dH=0$; $D_P$ is the field of the canonical equations of the mechanical system $(M,T_2,dU)$.

 Conversely, suppose that $D_P$ is a second order differential equation. Equation (\ref{Newton}) gives that $D_P\lrcorner\,\omega_2+dT+\alpha=0$ holds, where $\alpha$ is horizontal; since $D_P$ is an infinitesimal canonical transformation (Theorem \ref{tquantizado2}), $\alpha$ must be locally exact, then (locally) $\alpha=dU$ for some $U\in\C^\infty(M)$. Since $T$ is the Hamilton function associated with the tensor $(1/2)\Phi$, as before, quantizing results in
 $P = -(\hbar^2/2)\Delta+U$.
\end{proof}

The correspondence between differential operators $P$ and Hamiltonian fields $D_P$ in $T^*M$, established in Theorem \ref{tquantizado2}, is done in two steps:
\begin{align*}
&\text{Differential operator $P$}\quad\To\quad \text{Tensor $\Phi$, dequantized of $P$ by $\nabla$, according Theorem \ref{tquantizado1}}\\
&\text{Tensor $\Phi$}\quad\To\quad\text{Hamiltonian field, by means the symplectic structure $\omega_2$ of $T^*M$}
\end{align*}

By keeping the first step, the symplectic structure $\omega_2$ can be changed, thus changing the final field $D_P$.

For example, while the symplectic form $\omega_2$ makes the Laplacian $\Delta$ correspond to the geodesic field $D_G$, the introduction of an electromagnetic field $F_2$ (a 2-form closed in $M$), changing $\omega_2$ by the symplectic form $\omega_F:=\omega_2+F_2$, makes $\Delta$ correspond the field $D_L$ such that $D_L\lrcorner\omega_F+dT= 0$; $D_L$ is the field that governs the motion of virtual particles under the Lorentz force (adjusting the units of charge and mass); $D_L$ is a second order differential equation and a infinitesimal canonical transformation for $\omega_F$, not for $\omega_2$. The ``Hamilton-Jacobi equation'' for $D_L$ and $\omega_F$ is quantized as the Klein-Gordon equation. The details are given in \cite{QuantizacionMunozAlonso}.

\bigskip

\section{Quantization by means of Riemannian exponential}\label{dos}

In \cite{QuantizacionMunozAlonso} we have presented a quantization rule for the classical system $(M,T_2)$ by means of a linear symmetric connection $\nabla$ on $M$. That rule is defined from the geodesic field $D$ associated with the connection. Instead of use directly $\nabla$, we use the geodesic field $D$ of $\nabla$. The flow of the field $D$ on $TM$ allows us to establish an isomorphism of manifolds between a certain neighborhood $\mathcal U_x$ of vector $0$ in $T_xM$ and a neighborhood $U_x$ of $x$ in $M$, by associating with the vector $v_x\in\mathcal U_x$ the final point of the geodesic (curve solution of $D$) parameterized by $[0,1]$ that starts from the point $x$ with initial velocity $v_x$. When $x$ runs over $M$, the union of all the $\mathcal U_x$ is a neighborhood $\mathcal U$ of the 0-section in $TM$, and the flow of $D$ defines, in the above described way, a differentiable map $\textrm{exp}\colon\mathcal U\to M$, in which the 0-section of $TM$ is identified (as a part of $\mathcal U$) with $M$.

For each $f\in\A$, let
\begin{equation}\label{exp1}
\widehat f:=\textrm{exp}^*(f)\in\C^\infty(\mathcal U).
\end{equation}
Of $\widehat f$ the only thing we are interested in is its germ at the section $0$ of $TM$. If we denote by $\Ocal(M)$ the ring of germs of differentiable functions in neighborhoods of the section $0$ of $TM$, we identify $\widehat f$ with its germ $\in\Ocal(M)$. Thus, we have an injection of rings $\A\hookrightarrow\Ocal(M)$, $f\mapsto\widehat f$.

 In the injection of rings $\A\hookrightarrow\AT$ produced by the natural projection $TM\to M$, the $f\in\A$ give functions constant along the fibres, annihilated by each vertical differential operator on $TM$ (except these of order 0). But in the injection $f\mapsto\widehat f$, such operators no longer annihilate the $\widehat f$.

In Section \ref{intrometrica} we have seen how, with each covariant tensor field $a$ there is a vertical contravariant tensor field $\widehat\PHI_a$ on $TM$ associated by means of a rule determined by the metric $T_2$ and the symplectic form of $T^*M$ (or, alternatively, the Fourier transform). In local coordinates, $\widehat\PHI_a$ is obtained by substituting into the expression $a(x,dx)$, each $dx^j$ by the vertical vector field $-i\hbar g^{jk}\partial/\partial\dot x^k$. Or, by considering the coordinates $p_j$ as 1-forms, by substituting each $p_j$ by $-i\hbar\partial/\partial\dot x^j$ (\ref{corr2}).

The quantization rule given in \cite{QuantizacionMunozAlonso} is
\begin{defi}[Quantization by the exponential]
Let $(M,T_2)$ be a configuration space, $\nabla$ a symmetric linear connection on $M$, $\A\hookrightarrow\Ocal(M)$ ($f\mapsto\widehat f$) the injection determined
 by the exponential defined by the geodesic field $D$ of $\nabla$. For each symmetric covariant tensor field $a$ of order $r$, the differential operator $\widehat a$ quantized of the function $\dot a$ by $\nabla$ gives, for each $f\in\A$ the value
\be\label{agorro}
\widehat a(f):=\langle\widehat\PHI_a,d_0^r\widehat f\,\rangle
\ee
 where $\langle\,,\,\rangle$ is the tensor contraction and $d_0^r\widehat f$ is the $r$-th differential of $\widehat f$ along each fibre of $TM$, and taking the value at the 0 section.
\end{defi}

 The ``vertical differential'' $d^r\left(\widehat f|_{T_xM}\right)$ makes sense because $T_xM$ is a vector space.

 So as not to get lost in technicalities in the discussion that follows, let us suppose that the geodesic field $D$ of $\nabla$ is complete.
Let $\{\tau_s\}_{s\in\R}$ be the 1-parametric group of automorphisms of the manifold $TM$  generated by $D$. The exponential map is, in this case, the composition
\begin{equation}\label{diagramaexponencial}
\xymatrix{TM\ar[r]^-{\tau_1}\ar[dr]_-{\textrm{exp}} & TM\ar[d]^-{\pi\text{ (canonical projection)}}\\
  & M}
\end{equation}

From the mathematical point of view, the restriction of the parameter $s$ to the value $1$ is artificial. The natural thing is to consider an arbitrary segment $[0,s]$, $\tau_s$ instead of $\tau_1$, $\textrm{exp}_s=\pi\circ\tau_s$ and $\widehat f_s=\textrm{exp}_s^*(f)$. The classical magnitude $\dot a$ will be quantized as the  differential operator $\widehat a_s$:
\be\label{agorros}
\widehat a_s(f):=\langle\widehat\PHI_a,d_0^r\widehat f_s\,\rangle
\ee
The interesting thing is that the quantization rule $\dot a\to\widehat a_s$ changes in such a way that $\widehat a_s=s^r\widehat a$ for tensors of order $r$. Indeed, since $D$ is a field of the form (\ref{geodesico2}) (it does not matter how are the Christoffel symbols), it holds that $\pi\circ\tau_s(x,v_x)=\pi\circ\tau_1(x,sv_x)$ (the final point of the geodesic parameterized by $[0,s]$ with initial tangent vector $v_x$ at $x$, is the same that the final point of the geodesic parameterized by $[0,1]$ with initial tangent vector $sv_x$). It follows that $\widehat f_s(x,v_x)=\widehat f(x,sv_x)$, that is to say:
 $$\widehat f_s=\widehat f\circ (\text{Homothetie of ratio $s$ along each fibre of $TM$}).$$
 It is derived that $d_0^r\widehat f_s=s^rd_0^r\widehat f$, then $\widehat a_s=s^r\widehat a$. This means that quantization $\dot a\to\widehat a_s$ is deduced from quantization $\dot a\to\widehat a$ by replacing $h$ by $sh$. The field $D$ canonically produces a  1-parametric family of quantizations whose parameter is the Planck ``constant''.

\bigskip

\section{Identity of the two considered rules of quantization}\label{cuatro}
In this section all the functions are real.

Maintaining the above notation, $M$ is an smooth manifold of dimension $n$, $\nabla$ is a symmetric linear connection on $M$. The exponential map associated with $\nabla$ (defined on a neighborhood of $0$ on each fibre $T_{x_0}M$) assigns to each vector $v_{x_0}\in T_{x_0}M$ the point $\textrm{exp}(v_{x_0})\in M$ that is the final point of the geodesic of $\nabla$ parameterized by $[0,1]$ which starts from $x_0$ with tangent vector $v_{x_0}$. The local isomorphism $\exp\colon T_{x_0}M\to M$ assigns to each function $f\in\A$ a differentiable function defined in an neighborhood  of $0$ in $T_{x_0}M$; when $x_0$ runs over $M$, $f$ gives a function $\widehat f$ defined in a neighborhood of the 0 section of $TM$; the map $f\mapsto\widehat f$ injects $\A$ into the ring $\mathcal O(M)$ comprised by germs of smooth functions on neighborhoods of the 0-section of $TM$.

Whatever the local coordinates $x^1,\dots,x^n$ in an neighborhood of $x_0$, the corresponding $\dot x^1,\dots,\dot x^n$ are linear coordinates on $T_{x_0}M$. Thus, for each $g\in\C^\infty(T_{x_0}M)$ the following tensor is intrinsically defined
\be\label{identidad1}
d_0^rg=\sum_{j_1,\dots,j_r=1}^n\frac{\partial^rg}{\partial\dot x^{j_1}\cdots\partial\dot x^{j_r}}(0)\,d_0\dot x^{j_1}\cdots d_0\dot x^{j_r}
\ee

The local isomorphism $\exp\colon T_{x_0}M\to M$ gives an isomorphism of tangent spaces $T_0(T_{x_0}M)\simeq T_{x_0}M$ (the already known) makes to correspond
 to each vertical vector in $TM$ its geometric representative: $(\partial/\partial\dot x^j)_0\to(\partial/\partial x^j)_{x_0}$. The dual morphism makes to correspond $d_{x_0}x^j$ to $d_0\dot x^j$. This isomorphism transforms the $d_0^rg$ of (\ref{identidad1}) into a tensor at the point $x_0$ of $M$. In particular, for each $f\in\A$ we define the tensor
 \be\label{identidad2}
 d_{x_0}^rf:=\sum_{j_1,\dots,j_r=1}^n\frac{\partial^r\widehat f|_{T_{x_0}M}}{\partial \dot x^{j_1}\cdots\partial \dot x^{j_r}}(0)\,d_{x_0} x^{j_1}\cdots d_{x_0} x^{j_r}
 \ee
 which is the tensor at $x_0\in M$ that corresponds to $d_0^r\widehat f$ in (\ref{identidad1}) by means the isomorphism $T_{x_0}M\simeq T_0(T_{x_0}M)$.

 \begin{thm}\label{tidentidad1}
 For each $f\in\A$ and each $r$ is
 \be\label{identidad3}
 d_{x_0}^rf=\nabla^r_{x_0,\textrm{sym}}f
 \ee
 \end{thm}
  \begin{proof}
 The right and left members of equation (\ref{identidad3}) are symmetric tensors at $x_0\in M$. Therefore, it is sufficient to show that, for each vector $\xi\in T_{x_0}M$, when we apply both tensors on $\overset{r}{\overbrace{\xi\otimes\cdots\otimes\xi}}$ we get the same result.

 Let us denote by $x(s;\xi)$ the point of the geodesic curve which corresponds with the value $s$ of the parameter and the tangent vector $\xi$ at the initial point $x_0$: that is to say, $$x(0,\xi)=x_0,\quad \dot x(0,\xi)=\xi.$$ We assume, that vectors $\xi$ are contained in a neighborhood of $0\in T_{x_0}M$ where the exponential map is defined; the parameter $s$ runs over $[0,1]$ and $\textrm{exp}\,\xi=x(1,\xi)$.

 By a well known property of the geodesic curves, we have
 \begin{equation}\label{identidad4}
 x(s;\lambda\xi)=x(\lambda s;\xi),\qquad\text{for $|\lambda|\le 1$}.
 \end{equation}

  Let us fix the initial tangent vector, $\xi$, and define, for each covariant tensor field $a$ or order $r$ in $M$, the function
  \begin{equation}\label{identidad5}
  \alpha (s):=\langle a\, ,\,\overset{r}{\overbrace{\dot x(s;\xi)\otimes\cdots\otimes\dot x(s,\xi)}}\,\rangle,\qquad s\in[0,1]
  \end{equation}

  Along the geodesic curve, differentiation  with respect to the parameter $s$, corresponds with taking the covariant derivative with respect to the tangent vector field $\dot x(s,\xi)$. So that, taking into account that $\nabla_{\displaystyle\dot x}\,\dot x=0$, we get
  \begin{equation*}
    \alpha'(s) =\langle \nabla_{\displaystyle\dot x}\, a\, ,\,\overset{r}{\overbrace{\dot x\otimes\cdots\otimes\dot x}}\,\rangle_{x(s,\xi)}
  =\langle \nabla a\, ,\,\overset{r+1}{\overbrace{\dot x\otimes\cdots\otimes\dot x}}\,\rangle_{x(s,\xi)}
  \end{equation*}

  By induction on $r$, we obtain, for each given function $f\in\C^\infty(M)$,
  \begin{equation}\label{identidad6}
  %\begin{align*}
  \frac{d^rf(x(s,\xi))}{ds^r}=\langle \nabla^rf\, ,\, \dot x\otimes\cdots\otimes\dot x\rangle_{x(s,\xi)}
     =\langle \nabla^r_{\textrm{symm}}f\, ,\, \dot x\otimes\cdots\otimes\dot x\rangle_{x(s,\xi)}.
  %\end{align*}
    \end{equation}
 \smallskip

 On the other hand, we have $x(s;\xi)=x(1;s\, \xi)$ according (\ref{identidad4}) and hence
 $$f(x(s;\xi))=f(x(1;s\,\xi))=\widehat f(s\,\xi).$$
 By differentiating $r$ times and taking values at $s=0$ we arrive to
 \begin{equation}\label{identidad7}
    \left.\frac{d^rf(x(s,\xi))}{ds^r}\right|_{s=0}=
  \left.\frac{d^r\widehat f(s \,\xi)}{ds^r}\right|_{s=0}=\langle %\left.
  d_0^r\widehat f_{\,\,\mkern 1mu \vrule height 2ex\mkern2mu T_{x_0}M}\, ,\, \xi\otimes\cdots\otimes\xi\rangle.
    \end{equation}

    By taking into account Definition (\ref{identidad2}) and the canonical identification between $T_{x_0}M$ and $T_0(T_{x_0}M)$, formulas (\ref{identidad6}), (\ref{identidad7}) prove the theorem.
    \end{proof}

 \noindent\textbf{Conclusion.}
 The quantization defined in Section \ref{dos} (by means of the Riemannian exponential) and the one defined in Section \ref{tres} (by direct pairing between tensors) are identical. Indeed, the first one is obtained (in addition to the factors $-i\hbar$) for each tensor by applying on each $T_{x_0}M$ the direct quantization described in the Introduction by the formula (\ref{cuantizacionplana}) for $\R^n$, and applying to each
 $\widehat f_{\,\,\mkern 1mu \vrule height 2ex\mkern2mu T_{x_0}M}$; by (\ref{identidad1}), (\ref{identidad2}), (\ref{identidad3}), this is equivalent to directly coupling each  $\nabla^r_{x_0,\textrm{sym}}f$ with the given tensor.

 The exponential map reduces the symmetric covariant differential of arbitrary order to the corresponding linear differential along each fibre.

\bigskip

\section{Extension of the quantization rule to functions in $TM$ that are not Hamiltonian}\label{s:extension}

In this section, the configuration space is a manifold M endowed with a pseudo-Riemannian metric $T_2$ and a symmetric linear connection $\nabla$ in $TM$ (not necessarily the Levi-Civita connection of the metric).

Quantization rule (\ref{quantizado}) for Hamiltonians (polynomial functions in the fibers) cannot be extended to non-Hamiltonian functions in $TM$ by methods of  approximation. The reason, as we will see, is that (\ref{quantizado}) is associated with families of distributions in $TM$ with support in section 0, and no reasonable method of passing to the limit allows us to leave that section. However, the quantization rule (\ref{agorro}), equivalent to (\ref{quantizado}) in the case of Hamiltonians, extends to a very broad class of functions in $TM$, as we will see next.

Let $U$ be the neighborhood of section 0 of $TM$ in which is defined the Riemannian exponential given by $\nabla$. The map $\textrm{exp}\colon U\to M$ gives an injection of rings:
\[
\textrm{exp}^*\colon\C^\infty(M)\to\C^\infty(U),\qquad f\mapsto\widehat f:=\textrm{exp}^*(f)
\]

Let $\mathcal H$ be a complex vector space endowed with a locally convex topology. Each continuous linear map $\chi\colon\C^\infty(U)\to\mathcal H$ gives, composing with $\textrm{exp}^*$, a continuous linear map $$\widehat\chi\colon\C^\infty(M)\to\mathcal H.$$ In a very general sense, $\widehat\chi$ is the quantization of $\chi$ by means of the connection $\nabla$.
\medskip

Suppose that $\mathcal H$ is a space of functions in $M$ for which the functionals ``take values at $x\in M$'' are continuous (for example, $\mathcal H=\C^m(M)$ ($0\le m\le\infty$), $\mathcal{H}=\{$continuous functions bounded in $M$ with the norm of the supreme$\}$, etc.). For each $x\in M$, the composition of the map $\chi$ with ``take values at $x\in M$'' is a continuous linear map $\chi_{_{\displaystyle{_{x}}}}\colon\C^\infty(U)\to\mathbb{C}$, that is, a distribution with compact support in $U$. The restriction of $\chi_{_{\displaystyle{_{x}}}}$ to the sub-ring $\C^\infty(M)$ is a distribution $\widehat\chi_{_{\displaystyle{_{x}}}}$ in the manifold $M$.
\medskip

Let us detail the example of the quantized Hamiltonian operators. For a Hamiltonian $a\in\C^\infty(TM)$ the steps to arrive at the operator $\widehat a$ are:
\medskip

\noindent\textbf {1) (Depending only on the differentiable structure of $M$)}
\smallskip

\noindent The fiberwise Fourier transform from $T_xM$ to $T^*_xM$ changes the covariant tensor $a$ into a vertical contravariant tensor in $T^*M$ (Section \ref{previas}). In local coordinates
\[
a(x,dx)\quad\To\quad a(x,-i\hbar\,\partial/\partial p)
\]

\noindent\textbf {2) (Depending only on the metric $T_2$)}
\smallskip

\noindent The metric establishes an isomorphism of fibred bundles $TM\simeq T^*M$, transporting all the structures from one to the other. Thus, in local coordinates, the $p_k$' are, in $TM$, $p_k=g_{k\ell}\dot x^\ell$, identifiable with the covariant tensor $g_{k\ell}dx^\ell$ that, in turn, gives in $T^*M$ a vertical tangent field $-i\hbar g_{k\ell}\partial/\partial p_\ell$ that in $TM$ is $-i\hbar\partial/\partial\dot x^k$.
We have the correspondences
\[
a(x,dx)\quad\To\quad a(x,-i\hbar\partial/\partial p)\quad\To\quad \PHI_a(x,\partial/\partial\dot x),
\]
where $\PHI_a$ is obtained by substituting $\partial/\partial p_k$ by its traduction in coordinates of $TM$, that is $\partial/\partial p_k=g^{k\ell}\partial/\partial \dot x^\ell$.

The vertical contravariant tensor field in $TM$ canonically defines a differential operator $\widehat\PHI_a\colon\C^\infty(TM)\to\C^\infty(TM)$ that, if $r>0$, annihilates the subring $\C^\infty(M)$ (injected in $\AT$ by the canonical projection $TM\to M$); $\widehat\PHI_a$ is a ``vertical'' operator. For each $x\in M$, let us define in $TM$ the distribution $a_x$ that, on each $g\in\AT$ takes the value
\[
\langle a_x,g\rangle:=\widehat\PHI_a(g)(x,0)=\langle \delta|_{T_xM},\widehat\PHI_a(g)|_{T_xM}\rangle,
\]
where $\delta|_{T_xM}$ is the Dirac distribution at the origin of $T_xM$.

Calling $\widehat \PHI'_a|_{T_xM}$ the transposed of $\widehat \PHI_a|_{T_xM}$  with respect to the function-distribution duality in $T_xM$, it results
\[
a_x=\widehat \PHI'_a|_{T_xM}(\delta|_{T_xM})
.
\]

The distribution field $\{a_x\,\colon\, x\in M\}$ defines the operator
$\widehat{\widehat \PHI}_a\colon\AT\to\A$ as follows: 
$$(\widehat{\widehat \PHI}_ag)(x):=a_x(g),\qquad g\in\AT,\, x\in M.$$

\medskip

\noindent\textbf {3) (Depending on the connection $\nabla$)}
\smallskip

The quantized of $a$ is $\widehat a=\widehat{\widehat\PHI}_a\circ\textrm{exp}^*$; the field of distributions in $M$ is $\widehat a_x=a_x\circ\textrm{exp}^*$.

Distributions $a_x$ are supported in section 0 of $TM$. It is known (\cite{Schwartz}, p.100) that the distributions in $\R^n$ with support at the origin of coordinates are finite linear combinations of derivatives of  $\delta$.
For this reason, the fields of distributions in $TM$ with support in section 0 and ``vertical'' ($\chi_{_{\displaystyle{_{x}}}}$ kills the functions vanishing on $T_xM$) exclusively define differential operators, which are quantized of tensors.
\medskip

In order to extend the tensor quantization method to more general functions in $TM$ we must consider fields of  distributions $a_x$ with compact supports in each $U_x=U\cap T_xM$, that is, continuous linear functionals $a_x\colon\C^\infty(U)\to\mathbb{C}$. These distributions are, in each fiber $T_xM$, Fourier transforms of the functions capable of quantization. In what follows, and to avoid unpleasant discussions of details, we will assume that $\nabla$ is complete in the sense that any geodesic  can be extended to all values of the parameter, so $\textrm{exp}$ is defined in all of $TM$.

In such a case, the functions in $T_xM$ that give as Fourier transform distributions with compact support are characterized by a classic Paley-Wiener-Schwartz theorem (\cite{Schwartz}, p. 271 ff.).

Omitting easy to fill details in the proof, we can state:
\begin{thm}\label{t:extension}
Assuming that the connection $\nabla$ is geodesically complete, the functions in $\AT$ satisfying the Paley-Wiener-Schwartz conditions in the fibers are quantizable (by means of the exponential associated with $\nabla$) as applications $\A\to\A$. Those that are quantized as differential operators are precisely the Hamiltonians.
\end{thm}

Under additional assumptions, the quantization can be extended to functions more general than the theorem. For example, if $M$ is compact, $\textrm{exp}^*\A$ remains within the subspace of $\AT$ comprised by bounded functions which, in each fiber of $TM$, is restricted as a subspace of $\C^\infty(\R^n)$ whose dual contains non-compactly supported distributions; the functions in $\AT$ that, fiberwise give Fourier transforms in such a dual, are quantizable by $\textrm{exp}$.
\medskip

%%%%%%%%%%%%%%%%%%%%%%%%%%%%%%%%%%%%%%%%%%%

\end{document}